\documentclass[a4paper,fleqn]{cas-dc}
\usepackage[numbers]{natbib}
\usepackage{algorithm} 
\usepackage{algcompatible}
\usepackage{algpseudocode} 
\usepackage{caption} 
\usepackage{physics}
\usepackage{amsthm}
\floatname{algorithm}{Protocol}
\usepackage{xcolor}
\usepackage{float}
\usepackage{subcaption}
\usepackage{enumitem}
\usepackage{capt-of}   
\usepackage{needspace}
\definecolor{deepgreen}{rgb}{0, 0.6, 0.2}
\hypersetup{
  hidelinks,
  colorlinks=true,
  linkcolor=blue,
  citecolor=deepgreen
}

\let\STATE\State

\theoremstyle{definition}
\newtheorem{definition}{Definition}
\newtheorem{lemma}{Lemma}
\newtheorem{theorem}{Theorem}
\begin{document}
\let\WriteBookmarks\relax
\def\floatpagepagefraction{1}
\def\textpagefraction{.001}

\shorttitle{QPSO}    
\shortauthors{Zixian Gong et~al.}  
\title [mode = title]{Verifiable and Collusion-Resistant Multi-Party Quantum Private Set Operations
}

\author[1]{Zixian Gong}[orcid={0009-0005-7059-5040}]
\cormark[1]
\author[1]{Kun Tian}
\author[1]{Yi Zhang}
\author[2]{Fengxia Liu}
\ead{ArtoriasGong@ruc.edu.cn}

\affiliation[1]{organization={School of Mathematics, Renmin University of China},
                city={Beijing},
                postcode={100872}, 
                country={P. R. China}}

\affiliation[2]{organization={Great Bay University},
                city={DongGuan},
                postcode={523808}, 
                country={P. R. China}}
\cortext[1]{Corresponding author}

\begin{abstract}
Private set intersection (PSI) and, more broadly, private set operations (PSO) are fundamental primitives for secure multiparty computation (SMC), enabling participants to jointly compute set relations while revealing no information beyond the prescribed output. As quantum technologies advance, PSI have correspondingly evolved toward quantum secure phase. Existing quantum PSI (QPSI) solutions are limited in their threat models and collusion behavior between third party (TP) and participants. In this work, we present a multi-party QPSI (MP-QPSI) protocol that integrates verifiable quantum fully homomorphic encryption (vQFHE) as the verifiable outsourced quantum-evaluation layer and threshold fully homomorphic encryption (TFHE) as the threshold key-management mechanism. We instantiate the intersection computation via a \(C^{\mathsf{AND}}\) circuit accompanied by simulations on IBM Quantum Platform. We analyze correctness and participant privacy against TP, external eavesdroppers, and collusive behaviors, and we further prove verifiability against a malicious TP under the semantic security model. Finally, we present a modular framework perspective with several realizations, show how to extend the construction to quantum private set union (QPSU) via open-controlled operations. Compared with prior schemes, our protocol provides flexible set operations and stronger resilience under the TP model, including TP-participant collusion, thereby offering enhanced security and broader applicability.
\end{abstract}


\maketitle

\section{Introduction}
Private set intersection (PSI) is an important branch of secure multiparty computation (SMC), enabling parties to learn the intersection of their private sets while revealing no additional information beyond the intersection~\cite{Mea86,FNP04}. Owing to its strong privacy guarantees, PSI has become a key building block in a broad range of privacy-preserving applications, such as vertical federated learning~\cite{Angelou20}, contact discovery~\cite{Demmler18}, genomic data analysis~\cite{Shen18}.

As privacy requirements evolve, this primitive has been extended into a family of variants. Representative examples include private set union (PSU)~\cite{Vladimir19}, which enables parties to compute the union of private sets, as well as the cardinality-only variants PSI-CA and PSU-CA~\cite{Cristofaro12}, which reveal only the size of the intersection or union. Collectively, these and other functionality-oriented extensions for privacy-preserving set computation are often referred to as private set operations (PSO)~\cite{Kissner05}. These privacy-preserving set computation protocols are primarily instantiated via three mainstream cryptographic paradigms: homomorphic encryption (HE)~\cite{Chen17}, oblivious transfer (OT) and OT-extension, and oblivious pseudorandom functions (OPRF)~\cite{Pinkas14}. More recently, PSI constructions based on oblivious key-value stores (OKVS) have also gained significant traction~\cite{Pinkas20}. With the rapid progress of quantum technologies and advent of Shor’s algorithm~\cite{Shor95}, many cryptographic schemes based on number theoretical hard problems have become vulnerable. Beyond ongoing efforts on post-quantum cryptography (PQC) as an interim solution, a growing line of work turns to quantum cryptography and leverages properties of quantum mechanics to design quantum PSI (QPSI) protocols.

In 2015, \cite{QPSI_Shi15} proposed the first two-party QPSI protocol. However, \cite{Cheng16} later pointed out that it suffers from a fairness issue, and suggested introducing a trusted third party to remedy this limitation. Building on GHZ states, \cite{QPSI_Zhang20} extended the setting to three-parties and realized PSI-CA and PSU-CA under TP formulation, their protocol further provides resistance against collusion between the two participating parties. Subsequently, \cite{QPSI_M23} used single photons and unitary operations to extend QPSI to the multi-party setting (MP-QPSI), but their security guarantee only tolerates corruption of at most one participant, which is overly restrictive in large-scale setting. More recently, \cite{QPSI_Huang24} achieved MP-QPSI via rotation operations, and under the TP formulation, they addressed collusion among participants. Among all the proposed QPSO schemes, some of them have deception-sensitive techniques against a semi-honest TP~\cite{QPSI_Shi15, QPSI_Huang24}, their security guarantees typically rely on the assumption that the TP can not collude with any participant. \emph{This lack of TP-client collusion resistance constitutes a central motivation of our work}. To address it, we turn to another key cryptographic primitive that has advanced alongside quantum technologies, the fully HE (FHE).

After the early conceptual proposal~\cite{Rivest78} and the breakthrough construction of FHE~\cite{Gentry09}, with the extensive follow-up research and implementation in the classical setting, researchers naturally began to ask whether one could encrypt quantum data while still allowing a server to perform arbitrary quantum computations over the ciphertext. QHE was proposed firstly in 2014 \cite{Yu14}. In 2015, \cite{BJ15} introduced a quantum HE (QHE) scheme that the quantum information is encrypted via the quantum one-time pad (QOTP), while the corresponding QOTP keys are encrypted under a classical FHE scheme, enabling the TP to carry out quantum evaluation without learning the underlying plaintext. Because non-Clifford gates (in particular, the $\mathsf{T}$ gate) induce correction terms that depend nonlinearly on the QOTP masks and thus cannot be captured by Pauli $\textsf{X}/\textsf{Z}$ masks alone, \cite{Dulek16} used the garden-hose model~\cite{BFSS13} to construct quantum gadgets to remove the additional corrections. Although a variety of QHE constructions have since been developed,  unlike classical HE, they did not systematically incorporate verifiability of evaluation results until the introduction of verifiable QFHE (vQFHE) in~\cite{vQFHE17}.

\emph{Contribution}. In the QHE setting, the TP or server is typically assumed to possess strong quantum capabilities and can therefore execute relatively complex quantum circuits, and aims to enable such quantum computation while preventing the TP from learning any private information, \emph{it aligns naturally with the design objectives of QPSI}. Moreover, vQFHE can verify the correctness of the evaluated results. Building on this property, our contributions are as follows:
\begin{itemize}
    \setlength{\itemsep}{0pt}    
    \setlength{\parskip}{0pt}    
    \setlength{\topsep}{2pt}     
    \item We combine Threshold FHE (TFHE)~\cite{AJW11,Boneh18} with the vQFHE construction TrapTP from \cite{vQFHE17} to extend it to a multi-party setting. Specifically, in a model consisting of a powerful quantum TP, $n$ data-holding participants, and a trusted authority (TA) responsible for key generation, we instantiate the intersection computation via a $C^{\mathsf{AND}}$ circuit and simulated it on the IBM Quantum platform and thereby realize a MP-QPSI protocol.
    \item By leveraging the verifiability of vQFHE, our protocol achieves a key distinguishing property compared to existing QPSI schemes: it can use the intermediate classical information to efficiently detect whether the TP deviates from the \emph{prescribed target circuit}, thereby ensuring the correctness of the output and mitigating malicious behavior by the TP.
    \item Owing to TFHE’s threshold key sharing, our protocol our protocol provides collusion resistance that is absent from prior QPSI: as long as the number of corrupted parties is below the threshold $n_t$, \emph{neither collusion among participants nor collusion between the participants and the TP} can recover the secret key or learn the private inputs of honest parties including the final intersection information. Moreover, the structure improves availability in multi-party setting, allowing decryption and result recovery to proceed at least $n_t$ parties remain online, thereby tolerating temporary participants offline.
    \item We further demonstrate the extensibility of our construction: by replacing specific modules, the protocol can be adapted to realize different functionalities, such as reducing circuit complexity and enabling MP-QPSU via open-controlled operations.
\end{itemize}

\emph{Outline}. The remainder of this paper is organized as follows. Section~\ref{Pre} provides the necessary preliminaries. The proposed MP-QPSI protocol is presented in detail in Section~\ref{sPro}. Security analysis and proofs are given in Section~\ref{analysis}. A framework perspective and performance comparisons with related schemes are provided in Sections~\ref{Framework} and~\ref{Performance}, respectively. Finally, Section~\ref{Conclusion} will give a conclusion.
\section{Preliminaries}\label{Pre}
\subsection{Threshold FHE \& Secret Sharing}\label{TFHESS}
\begin{definition}[TFHE]\cite{AJW11}\label{TFHE}
For a set of $n$ parties $\{P_i\}_{i\in[n]}$ with threshold $t\le n$. A $t$-out-of-$n$ $\mathsf{TFHE}$ scheme is a 4-tuple of algorithms $(\mathsf{TFHE.KeyGen},\mathsf{TFHE.Enc},\mathsf{TFHE.Ev}\\\mathsf{al},\mathsf{TFHE.Dec})$ defined as follows:
\begin{itemize}
  \setlength{\itemsep}{0pt}    
  \setlength{\parskip}{0pt}    
  \setlength{\topsep}{2pt}     
  \item $\mathsf{TFHE.KeyGen}(1^\kappa)$: Take as input the security parameter $\kappa$, each party $P_i$ obtains a tubple of keys including the common public key $pk$, evaluation key $evk$ and a private share $sk_i$ of the secret key $sk$.
  \item $\mathsf{TFHE.Enc}_{pk}(\mu)$: Input the public key $pk$ and a plaintext in message space $\mu\in\mathcal{M}$, output a ciphertext $c$.
  \item $\mathsf{TFHE.Eval}_{evk}(f,c_1,\dots,c_\ell)$: For some ciphertexts $c_1,\\\dots,c_\ell$ and a bounded boolean circuit
  $f$, this deterministic poly-time algorithm uses $evk$ and output an evaluated ciphertext $c'$.
  \item $\mathsf{TFHE.Dec}_{\{sk_i\}_{i \in S}}(c')$: Given an evaluated ciphertext $c'$ and the $\{sk_i\}$ of any subset $S\subseteq [n]$ with $|S|\ge t$, the parties in $S$ jointly decrypt $c'$ to obtain a plaintext $u' = f(u_1,\dots,u_\ell)$.
\end{itemize}
\end{definition}
\begin{definition}[TFHE Evaluation Correctness]\label{TFHE.Correctness}
For a $t$-out-of-$n$ $\mathsf{TFHE}$ scheme as in Definition~\ref{TFHE}. We say that it satisfies evaluation \emph{correctness} if for all security parameters $\kappa$, all circuits $f$ in the supported class, and every subset $S\subseteq[n]$ with $|S|\ge t$. Let $(pk,\{sk_i\}_{i\in[n]})\leftarrow\mathsf{TFHE.KeyGen}(1^\kappa)$ and $c_i\leftarrow\mathsf{TFHE.Enc}_{pk}(\mu_i)\ (i\in[\ell])$, $c'\leftarrow\mathsf{TFHE.Eval}_{evk} (f,c_1,\dots,c_\ell)$,  the following holds:
\[
  \Pr\big[\mathsf{TFHE.Dec}_{\{sk_i\}_{i \in S}}(c') = f(\mu_1,\dots,\mu_\ell)\big]
  = 1-\mathsf{negl}(\kappa).
\]
\end{definition}

\begin{definition}[Threshold SS]\cite{Boneh18}\label{ss}
Let $P=\{P_i\}_{i\in[n]}$ be a set of parties and $t\le n$. A $(t,n)$-threshold secret sharing scheme for secret space $\mathcal{K}$ is a pair of algorithms $\mathsf{SS}=(\mathsf{SS.Share},\mathsf{SS.Combine})$:
\begin{itemize}
  \setlength{\itemsep}{0pt}
  \setlength{\parskip}{0pt}
  \setlength{\topsep}{2pt}
  \item $\mathsf{SS.Share}(k)\to (s_1,\dots,s_n)$:
  On input a secret $k\in\mathcal{K}$ outputs a share $s_i$ for each party $P_i$.
  \item $\mathsf{SS.Combine}(\{s_i\}_{i\in S})\to k$:
  On input the shares of a subset $S\subseteq[n]$ for $|S|\ge t$ outputs either a secret $k\in\mathcal{K}$ or $\perp$.
\end{itemize}
\end{definition}
\begin{definition}[SS Correctness]\label{SS.Correctness}
A $(t,n)$-threshold secret sharing scheme in Definition~\ref{ss} is \emph{correct} if for all $k\in\mathcal{K}$ and all $S\subseteq[n]$ with $|S|\ge t$, letting  
$(s_1,\dots,s_n)\leftarrow\mathsf{SS.Share}(k)$ we have
\[
  \Pr\bigl[\mathsf{SS.Combine}(\{s_i\}_{i\in S}) = k\bigl] = 1 .
\]
\end{definition}

\begin{definition}[SS Privacy]\cite{Shamir79}\label{SS.Privacy}
A $(t,n)$-threshold secret sharing scheme has \emph{$t$-privacy} if  
for all $S\subseteq[n]$ with $|S|\le t-1$, any $k_0,k_1\in\mathcal{K}$, and
independent randomness, if
$(s_{b,1},\dots,s_{b,n})\leftarrow\mathsf{SS.Share}(k_b)$ for $b\in\{0,1\}$, then
\[
  \{s_{0,i}\}_{i\in S} \approx \{s_{1,i}\}_{i\in S}.
\]
\end{definition}
\begin{lemma}\label{tfhe-basic}
For a $\mathsf{TFHE}$ scheme whose secret key $sk$ is shared by a $(t,n)$-$\mathsf{SS}$ scheme, the following hold:
\begin{itemize}
  \setlength{\itemsep}{0pt}
  \setlength{\parskip}{0pt}
  \setlength{\topsep}{2pt}
  \item Any group of $t$ or more out of the $n$ participants can collectively reconstruct the private key $sk$;
  \item Any group of at most $t-1$ participants cannot reconstruct the private key $sk$;
  \item Ciphertexts under $\mathsf{TFHE}$ can be homomorphically evaluated via $\mathsf{TFHE.Eval}$;
  \item Any ciphertext can be correctly decrypted by any qualified set $S$ with $|S|\ge t$.
\end{itemize}
\end{lemma}

\subsection{QHE\,\cite{BJ15}}\label{SQHE}
\begin{definition}[QHE]\label{QHE}
Let $\kappa$ be the security parameter.
A quantum homomorphic encryption scheme $\mathsf{QHE}$ for message space $\mathcal{M}$ and cipherspace $\mathcal{C}$ is a 4-tuple of quantum polynomial-time (QPT) algorithms $(\mathsf{QHE.KeyGen},\mathsf{QHE.Enc},\\\mathsf{QHE.Eval},\mathsf{QHE.Dec})$
defined as follows:
\begin{itemize}
  \setlength{\itemsep}{0pt}
  \setlength{\parskip}{0pt}
  \setlength{\topsep}{2pt}
  \item $\mathsf{QHE.KeyGen}(1^\kappa)$: On input $1^\kappa$ outputs a classical public key $pk$, a classical secret key $sk$, and a quantum evaluation key $\rho_{\mathsf{evk}}\in D(\mathcal{R}_{\mathsf{evk}})$.
  \item $\mathsf{QHE.Enc}_{pk}(\rho)$: Through $pk$, this quantum channel maps a message state $\rho\in D(\mathcal{M})$ to a ciphertext state $\sigma\in D(\mathcal{C})$.
  \item $\mathsf{QHE.Eval}_{\rho_{evk}}^{\mathsf{C}}(\sigma)$: For every quantum circuit $\mathsf{C}$ with $\Phi_{\mathsf{C}} : D(\mathcal{M}^{\otimes n})\to D(\mathcal{M}^{\otimes m})$, $\mathsf{QHE.Eval}_{\rho_{evk}}^{\mathsf{C}}$ uses the evaluation key $\rho_{\mathsf{evk}}$ and maps $n$ ciphertext registers $\sigma$ to $m$ ciphertext registers $\sigma'$ as $D(\mathcal{R}_{\mathsf{evk}}\otimes\mathcal{C}^{\otimes n}) \to D(\mathcal{C}'^{\otimes m})$.
  \item $\mathsf{QHE.Dec}_{sk}(\sigma')$: For every $sk$, the quantum channel
  $\mathsf{QHE.Dec}_{sk}$ maps the ciphertext state $\sigma'$ in $D(\mathcal{C}')$ back to a plaintext state $\rho'$ in $D(\mathcal{M})$.
\end{itemize}
\end{definition}
\begin{definition}[QHE Evaluation Correctness]\label{QHE.Correctness}
A $\mathsf{QHE}$ scheme is \emph{correctness} if for every efficient quantum circuit $\mathsf{C}$ with induced channel $\Phi_\mathsf{C}$ and every input state
$\rho\in D(\mathcal{M}^{\otimes n})$,
\[
  \Bigl\|
    \mathsf{QHE.Dec}_{sk}^{\otimes m}
    \bigl(\mathsf{QHE.Eval}^{\mathsf{C}}_{\rho_{\mathsf{evk}}}
          (\mathsf{QHE.Enc}_{pk}^{\otimes n}(\rho))\bigr)
    - \Phi_{\mathsf{C}}(\rho)
  \Bigr\|_{\mathrm{tr}}
\]
is at most $\mathsf{negl}(\kappa)$.
\end{definition}

\subsection{vQFHE\,\cite{vQFHE17}}\label{svQFHE}
\begin{definition}[vQFHE]\label{vQFHE}
Let $\kappa$ be the security parameter. A $\mathsf{vQFHE}$ scheme is a set of QPT algorithms $(\mathsf{vQFHE.KeyGen},\\ \mathsf{vQFHE.Enc}, \mathsf{vQFHE.Eval},\mathsf{vQFHE.VerDec})$ defined as follows:
\begin{itemize}
  \setlength{\itemsep}{0pt}
  \setlength{\parskip}{0pt}
  \setlength{\topsep}{2pt}
  \item $\mathsf{vQFHE.KeyGen}(1^\kappa)$ and $\mathsf{vQFHE.Enc}_{pk}(\rho)$ have the same syntax as $\mathsf{QHE.KeyGen}$ and $\mathsf{QHE.Enc}_{pk}$ in Definition~\ref{QHE}.
  \item $\mathsf{vQFHE.Eval}^{\mathsf{C}}_{\rho_{\mathsf{evk}}}(\sigma)$: On input a quantum circuit $\mathsf{C}$, $\mathsf{vQHE.Eval}_{\rho_{evk}}^{\mathsf{C}}$ uses the evaluation key $\rho_{\mathsf{evk}}$ and maps $\sigma$ to $\sigma'$ with an extra classical computation $log$ output as $D(\mathcal{R}_{\mathsf{evk}}\otimes\mathcal{C}^{\otimes n}) \to D(\mathcal{L}\otimes\mathcal{C}'^{\otimes m})$. 
  \item $\mathsf{vQFHE.VerDec}_{sk}(\mathsf{C},log,\sigma')$:
  Using the secret key $sk$, a circuit $\mathsf{C}$, a computation $log$, and a ciphertext state $\sigma'$, the algorithm outputs a pair $(\rho',b)$, where $\rho'$ is the decrypted state and $b\in\{\mathsf{acc},\mathsf{rej}\}$ is the verification result.
\end{itemize}
\end{definition}
\begin{definition}[vQFHE Evaluation Correctness]\label{vQFHE.Correctness}
A $\mathsf{vQFHE}$ scheme is \emph{correctness} if for every security parameter $\kappa$, every key triple $(pk,sk,\rho_{\mathsf{evk}})\leftarrow\mathsf{vQFHE.KeyGen}(1^\kappa)$ and every poly-size quantum circuit $\mathsf{C}$ acting on plaintext, the channel induced by honest execution satisfies
\[
  \mathsf{vQFHE.VerDec}^{\mathsf{C}}_{sk}
  \circ
  \mathsf{vQFHE.Eval}^{\mathsf{C}}_{\rho_{\mathsf{evk}}}
  \circ
  \mathsf{vQFHE.Enc}_{pk}
  =
  \Phi_{\mathsf{C}}
\]
with overwhelming probability $1-\mathsf{negl}(\kappa)$.
\end{definition}

\begin{definition}[$\kappa$-SEM-VER]\label{sem-ver}
As to a $\mathsf{vQFHE}$ scheme, for any QPT adversary $\mathcal{A}$ which manipulates a ciphertext (and side info), $\mathcal{A}$ outputs a modified ciphertext together with a circuit description $\mathsf{C}$ and a computation $log$, the  \emph{real} channel $\Phi_\mathcal{A}$ is defined as follows:
\[
  \mathsf{Real}^{\mathcal{A}}
  \ :=\ 
  \mathsf{vQFHE.VerDec}^{\mathsf{C}}_{sk}\ \circ\ \mathcal{A}\ \circ\ \mathsf{vQFHE.Enc}_{pk}.
\]
For any QPT simulator $\mathcal{S}$ that never sees a ciphertext but
only declares a circuit and an $\mathsf{acc/rej}$ decision, this defines the
\emph{ideal} channel $\Phi_\mathcal{S}$:
\[
  \mathsf{Ideal}^{\mathcal{S}}
  \ :=\ 
  \mathsf{ctrl\text{-}\oslash}\ \circ\ \Phi_{\mathsf{C}}\ \circ\ \mathcal{S},
\]
where $\Phi_{\mathsf{C}}$ is the channel implemented on the plaintext directly by the circuit $\mathsf{C}$, and $\mathsf{ctrl\text{-}\oslash}$ replaces the output by a fixed state whenever the decision is "$\mathsf{rej}$".

The $\mathsf{vQFHE}$ is \emph{semantically $\kappa$-verifiable} if for every QPT adversary $\mathcal{A}$, there exists a QPT simulator $\mathcal{S}$ such that for all QPT message generators $\mathcal{M}$ and distinguishers
$\mathcal{D}$,
\[
\Bigl|
  \Pr\bigl[\mathcal{D}(\mathsf{Real}^{\mathcal{A}}(\mathcal{M}(\rho_{\mathsf{evk}}))) = 1\bigr]
  -
  \Pr\bigl[\mathcal{D}(\mathsf{Ideal}^{\mathcal{S}}(\mathcal{M}(\rho_{\mathsf{evk}}))) = 1\bigr]
\Bigr|
\]
with negligible probability $\mathsf{negl}(\kappa)$.
\end{definition}

\subsection{QOTP\,\cite{qotp,qotprograms}}\label{SQOTP}
\begin{definition}[QOTP]\label{QOTP}
For a single-qubit state $\rho$ and a classical key
$(a,b)\in\{0,1\}^2$, the quantum one-time pad masks $\rho$ by Pauli operators
\[
  \textrm{QOTP}_{a,b}(\rho) \;=\; \mathsf{X}^{a} \mathsf{Z}^{b}\,\rho\,\mathsf{Z}^{b} \mathsf{X}^{a}.
\]
Using the same key $(a,b)$ again removes the mask, so the plaintext is
easily recovered. If $(a,b)$ is chosen uniformly at random, then for all
$\rho$,
\[
  \frac{1}{4}\sum_{a,b\in\{0,1\}} \mathsf{X}^{a} \mathsf{Z}^{b}\,\rho\,\mathsf{Z}^{b} \mathsf{X}^{a}
  \;=\; \frac{\mathbb{I}_2}{2},
\]
This property is used in the $\mathsf{QHE.Enc}$ that the encrypted state is maximally mixed from the adversary’s point of view.
\end{definition}
\subsection{MAC\,\cite{Broadbent16}}\label{SMAC}
\begin{definition}[MAC]\label{MAC}
Let $\mathcal{M}$ be a message space and $\mathcal{T}$ be a tag space.
A MAC is a triple of algorithms $\mathsf{MAC}=(\mathsf{MAC.KeyGen},\mathsf{MAC.Sign},\mathsf{MAC.Ver})$ defined as follows:
\begin{itemize}
  \setlength{\itemsep}{0pt}
  \setlength{\parskip}{0pt}
  \setlength{\topsep}{2pt}
  \item $\mathsf{MAC.KeyGen}(1^\kappa)$:
        Outputs a MAC key $k$.
  \item $\mathsf{MAC.Sign}_k(m)$:
        Using the MAC key $k$, for a message $m\in\mathcal{M}$ returns a tag pair $(m,\tau)$ where $\tau\in\mathcal{T}$.
  \item $\mathsf{MAC.Ver}_k(m,\tau)$:
    For $(m,\tau)\in\mathcal{M}\times\mathcal{T}$ returns a bit
        $b\in\{0,1\}$ indicating whether $\tau$ is a valid tag on $m$
        under key $k$.
\end{itemize}
\end{definition}
\begin{definition}[MAC Correctness]\label{MAC.correctness}
For a $\mathsf{MAC}$ scheme as defined in Definition~\ref{MAC}. We say that $\mathsf{MAC}$
satisfies correctness if for every security parameter $\kappa$ and every message $m\in\mathcal{M}$, $k \leftarrow \mathsf{MAC.KeyGen}(1^\kappa)$ and $\tau \leftarrow \mathsf{MAC.Sign}_k(m)$:
\[
  \Pr\bigl[\,b=1 :\; b \leftarrow \mathsf{MAC.Ver}_k(m,\tau)
  \bigr] = 1 .
\]
\end{definition}
\begin{definition}[MAC Security]\label{MAC.Security}
$\mathsf{MAC}$ is \emph{unconditional one-time secure} if for every probabilistic algorithms $\mathcal{A}$:
\[\Pr\bigl[
      \mathsf{MAC.Ver}_k(m,\tau)=1 :
      k \xleftarrow{\$} \{0,1\}^\kappa,\,
      (m,\tau) \leftarrow \mathcal{A}(1^\kappa)
  \bigr]
\]
is negligible in a security parameter $\kappa$.
\end{definition}

\begin{definition}[MAC EUF-CMA Security]\label{def:MAC-EUF-CMA}
Let $\mathsf{MAC}=(\mathsf{MAC.KeyGen},\mathsf{MAC.Sign},\mathsf{MAC.Ver})$
be a MAC for message space $\mathcal{M}$. Consider the following
experiment with an adversary $\mathcal{A}$:

\medskip
\noindent
\textbf{Experiment} $\mathsf{EUF\mbox{-}CMA}^{\mathcal{A}}_{\mathsf{MAC}}(\kappa)$
\begin{enumerate}
  \setlength{\itemsep}{0pt}
  \setlength{\parskip}{0pt}
  \setlength{\topsep}{2pt}
  \item Sample $k \leftarrow \mathsf{MAC.KeyGen}(1^\kappa)$.
  \item Give $\mathcal{A}$ oracle access to
        $\mathcal{O}_{\mathsf{sign}}(m)=\mathsf{MAC.Sign}_k(m)$.
        Let $Q$ be the set of messages queried to this oracle.
  \item Eventually $\mathcal{A}$ outputs a pair $(m^\star,\tau^\star)$.
        The experiment outputs $1$ (``$\mathcal{A}$ wins'') iff
        $\mathsf{MAC.Ver}_k(m^\star,\tau^\star)=1$ and
        $m^\star\notin Q$.
\end{enumerate}

We say that $\mathsf{MAC}$ is \emph{existentially unforgeable under
adaptive chosen-message attacks} (EUF-CMA secure) if for every PPT
adversary $\mathcal{A}$,
\[
  \Pr\bigl[\,\mathsf{EUF\mbox{-}CMA}^{\mathcal{A}}_{\mathsf{MAC}}(\kappa)=1\,\bigr]
  \le \mathsf{negl}(\kappa) .
\]
\end{definition}

\section{Proposed Protocol}\label{sPro}
\begin{figure*}[t]  
    \centering
    \includegraphics[width=0.8\linewidth]{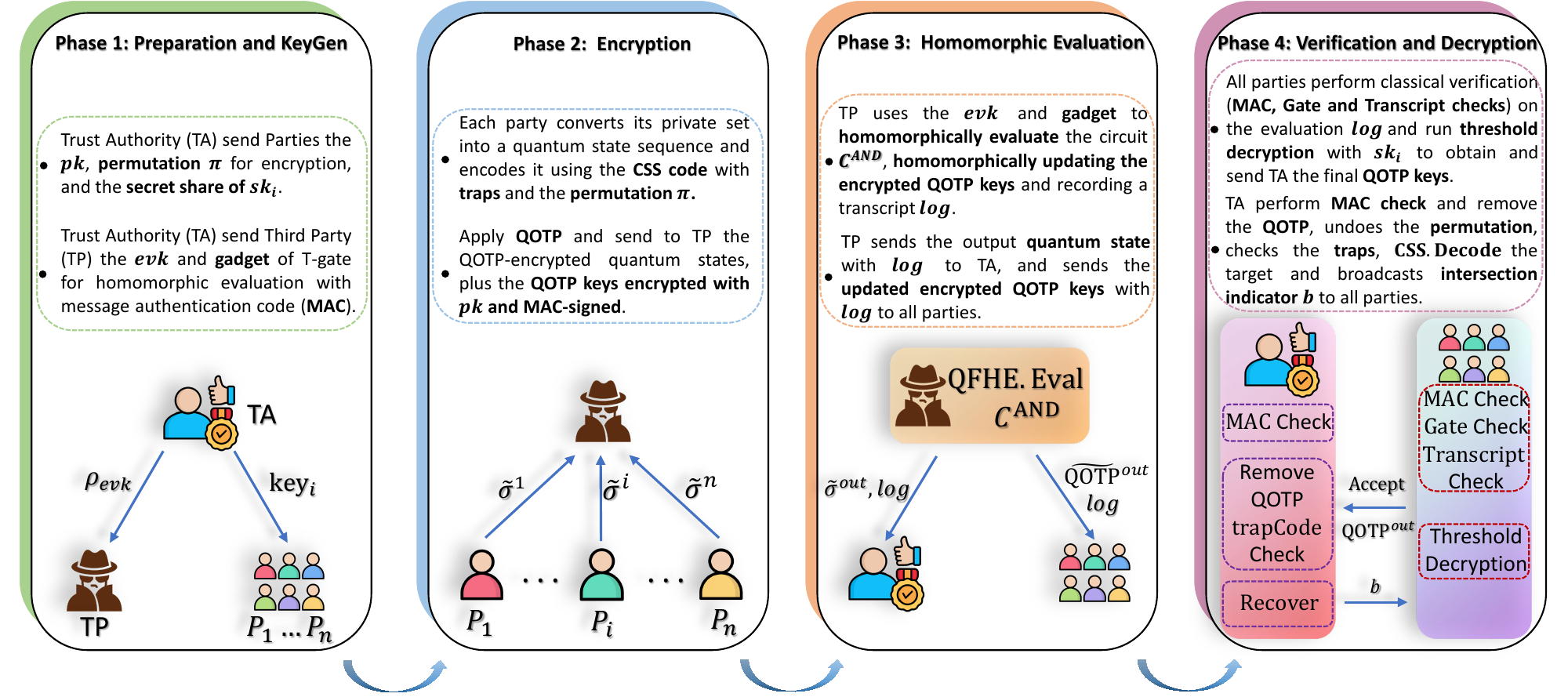}
    \caption{Overview of the proposed MP-QPSI protocol.}
    \label{overview}
\end{figure*}
\noindent\textbf{Protocol Overview.} In this section we explain how TrapTP-based vQFHE, combined with a classical threshold FHE scheme, can be lifted to the multi-party setting and instantiated in our MP-QPSI protocol in Protocol~\ref{MP-QPSI}.  For ease of exposition, we first provide a high-level flowchart of the whole protocol in Fig.~\ref{overview}.  As illustrated there, the protocol proceeds in four phases. From a capability perspective, we model the TP as a powerful quantum server capable of running large-scale quantum computations
and the associated classical homomorphic updates. TA is assumed to have only moderate quantum resources, sufficient for preparing a small number of authenticated auxiliary tools for evaluation while each participant \(P_i\) is a lightweight client equipped with very limited quantum capability, only needed to execute a fixed shallow encryption circuit; all remaining operations on the participants' side are purely classical. The concrete actions of TP, TA and the participants in each phase will be detailed in the following content.\\

\noindent\textbf{Preparation and KeyGen.} In the preparation phase, only the TA is active.  
TA first initializes the classical \textsf{TFHE} scheme and derives the evaluation keys, which will later be used by TP to homomorphically update QOTP ciphertexts and sampling a permutation \(\pi\) to hide the traps. Using the TrapTP \textsf{GadgetGen} algorithm~\cite{vQFHE17} to prepare \textsf{T}-gate gadgets with \textsf{MagicPool} for non-clifford gates in the quantum circuits. TA also generates an auxiliary target block with QOTP mask and signs all classical information above with a fixed EUF-CMA secure MAC \cite{MAC20}. Finally, TA prepares the client-side classical keys by choosing the base encryption key \(pk_0\), \textsf{SS.Share} the final decryption key into shares \(\{sk_i^{t}\}\). For the whole protocol, we use $\,\widetilde \cdot \,$ to denote the encryptions of both classical and quantum information.\\
\begin{figure*}[t]  
    \centering
    \includegraphics[width=0.6\linewidth]{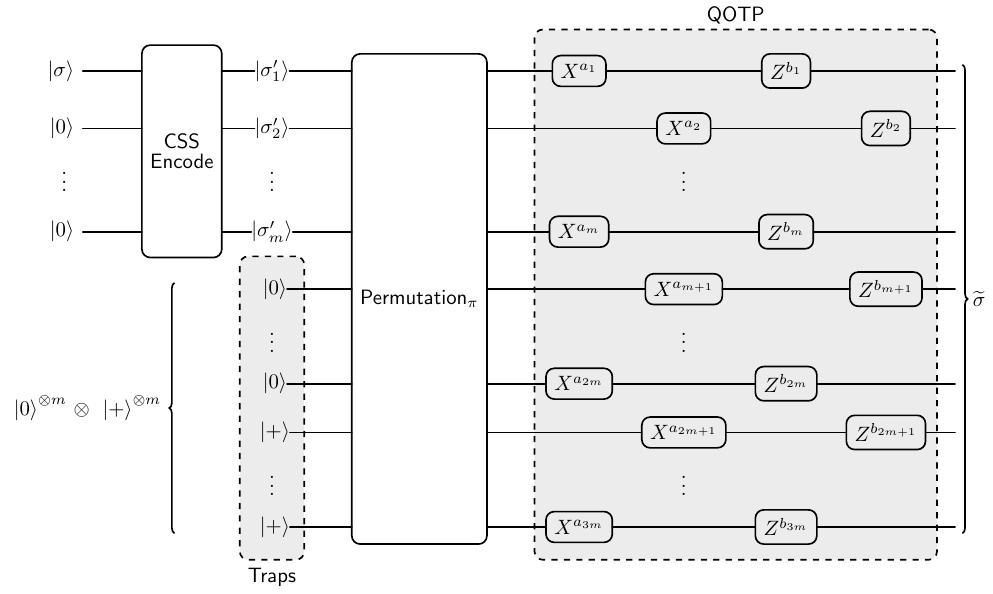}
    \caption{Quantum Circuit of Encryption.}
    \label{circuit1}
\end{figure*}

\noindent\textbf{Encryption.} In the encryption phase, all participants \(P_i\) are active as depicted in Fig.~\ref{circuit1}.  Following TrapTP~\cite{vQFHE17}, we fix a self-dual \([[m,1,d]]\) \textsf{CSS} code (so that \textsf{H} and \textsf{CNOT} are transversal), which can correct \(\kappa = d_c\) errors where \(d = 2d_c+1\) and block size \(m = \mathrm{poly}(d)\) chosen sufficiently large.  Each party \(P_i\) first sets its private set \(\mathcal{S}_i\) into a uniform-length quantum sequence \(\mathcal{S}_i^Q\) according to a common ordering of the universe.  For every logical qubit, \(P_i\) applies \textsf{CSS} encoding, appends \(m\) \(|0\rangle\)-traps and \(m\) \(|+\rangle\)-traps, and hides the trap positions using the permutation \(\pi\) (in principle, each party could use its own permutation \(\pi_i\), at the cost of a more involved classical homomorphic update at TP).  The resulting \(3m\)-qubit block is then masked by QOTP, the corresponding QOTP keys are encrypted and MAC-signed using the local key \(k_i\), and all these ciphertexts are sent to TP.\\

\noindent\textbf{Homomorphic Evaluation.} In this phase only TP is active.  After receiving the authenticated evaluation key package \(\rho_{\mathsf{evk}}\) from TA and the encrypted quantum inputs \(\widetilde{\mathcal{S}}_i^{Q}\) from all participants, TP homomorphically evaluates the fixed multi-party AND
circuit \(C^{\mathsf{AND}}\) on these ciphertexts, following the TrapTP evaluation procedure~\cite{vQFHE17} and the standard QHE paradigm~\cite{Dulek16,BJ15}.  

The circuit is viewed as an alternating sequence of Clifford layers \(C^{(\ell)}_{\mathsf{Cliff}}\) and non-Clifford \(\mathsf{T}\)-layers \(C^{(\ell)}_{\mathsf{T}}\).  For a
Clifford layer, the relation \(\mathsf{C}\, \mathsf{X}^{a}\mathsf{Z}^{b}= \mathsf{X}^{a'}\mathsf{Z}^{b'} \mathsf{C}\) allows TP to use \(\mathsf{HE.Eval}\) to homomorphically update the encrypted QOTP keys \((\widetilde a,\widetilde b)\) to \((\widetilde {a'},\widetilde {b'})\) in the classical ciphertext space. When entering a non-Clifford \(\mathsf{T}\)-layer, the relation \(\mathsf{T}\mathsf{X}^{a}\mathsf{Z}^{b} = \mathsf{P}^{a}\mathsf{X}^{a}\mathsf{Z}^{b}\mathsf{T}\) introduces an extra phase \(\mathsf{P}^{a}\) that depends on the secret key bit \(a\).  Since TP is not allowed to learn \(a\) and only holds its encryption \(\widetilde a\), Thorough "Gadget" thought in \cite{GC99}, TP removes this unwanted phase by consuming one gadget generated by \textsf{GadgetGen} based on the garden-hose model~\cite{BFSS13}. At the same time, the corresponding classical QOTP-key ciphertexts are homomorphically \emph{recrypted}, so that after processing all \(t\) \(\mathsf{T}\)-gates the final keys are encrypted under the last-level public key \(pk_t\).  To maintain authentication security in the sense of~\cite{Broadbent16}, due to the introduction of the trap code, permutation and \textsf{CSS} structure, at each logical gate application it is necessary to perform
appropriate logical measurements, and to use the evaluation key together with the encrypted permutation \(\widetilde{\pi}\) to homomorphically evaluate the corresponding \(\mathsf{unpermute}\) and \(\mathsf{permute}\) operations; we refer to~\cite{vQFHE17,qotprograms}
for the detailed procedure and for a comparison of the underlying trap-code constructions.

After \(C^{\mathsf{AND}}\) have been processed, the logical multi-party AND has been applied to the prepared auxiliary block \(\widetilde\sigma_{\mathsf{aux}}\), yielding a output block \(\widetilde\sigma_{\mathsf{out}}\) and the updated encrypted QOTP keys \((\widetilde a^{\,\mathsf{out}},\widetilde b^{\mathsf{out}})\) under \(pk_t\).  TP appends all intermediate classical information—the FHE evaluation transcript, the sequence of evaluated gates to a log that will later be used for verification.
Finally, TP sends \((\widetilde a^{\,\mathsf{out}},\widetilde b^{\mathsf{out}}, \log,C^{\mathsf{AND}})\) to all participants \(P_i\), and, due to the no-cloning property of quantum information, forwards the unique quantum output block \(\widetilde\sigma_{\mathsf{out}}\) together with the same log to TA.\\

\noindent\textbf{Verification and Decryption.} In the final phase, the parties jointly run a two-layer verification and decryption procedure, consisting of  \emph{Classical.VerDec} and \emph{Quantum.VerDec}.  On the classical side, all participants \(P_i\) perform MAC check, gate array check and transcript (log) check and apply threshold decryption on the final QOTP keys while TA applies trap-based check on the quantum side and proceeds with the final decryption to recover the intersection result.

\begin{algorithm*}[!tp]
  \caption{MP-QPSI}\label{MP-QPSI}
  \begin{algorithmic}[1]
  \STATE \textbf{Protocol Input:} Security parameter $\kappa$, the upper bound on the number of $\mathsf{T}$ and $\mathsf{H}$ gates $t,h\in \mathbb{N}$, number of parties $n$, and a fixed self-dual CSS code $\mathsf{CSS}\,[[m,1,d]]$.
    \STATE \textbf{Protocol Output:} $\perp$ if protocol abort, else \{0,1\} represent the intersection output.
    \STATE \textbf{Phase 1: Preparation and KeyGen}
    \STATE For Trust Authority (TA):
    \STATE \hspace{0.6em} Run $k_{\mathsf{TA}}\leftarrow \mathsf{MAC.KeyGen}(1^\kappa)$ to obtain the MAC key for TA.
    \STATE \hspace{0.6em} Uniformly sample a permutation $\pi \leftarrow_R S_{3m}$ from the permutation group on $3m$ positions.
    \STATE \hspace{0.6em} For each $i \in \{0,\dots,t\}$, run
           $(sk_i,pk_i,evk_i) \leftarrow \mathsf{HE.KeyGen}(1^\kappa)$.
    \STATE \hspace{0.6em} Set
           $\mathsf{KEY}_{eval} := (evk_0,\dots,evk_t,pk_0,\dots,pk_t)$ and compute
           $\widetilde{\pi} \leftarrow \mathsf{HE.Enc}_{pk_0}(\pi)$.
    \STATE \hspace{0.6em} Prepare the target state as $\sigma_{\mathsf{aux}}$, sample QOTP keys $(a_{0},b_{0})\leftarrow_R \{0,1\}$, then encrypt $\sigma_{\mathsf{aux}}$ from QPTP as: $\sigma'_{\mathsf{aux}} \leftarrow \mathsf{X}^{a_{0}}\mathsf{Z}^{b_{0}}\sigma_{\mathsf{aux}}\mathsf{X}^{a_{0}}\mathsf{Z}^{b_{0}}$ and return \[\widetilde{\sigma}_{\mathsf{aux}} \leftarrow \sigma'_{\mathsf{aux}}\otimes \mathsf{MAC.Sign}_{k_{TA}}(\mathsf{HE.Enc}_{pk_0}(a_0,b_0)).\]    
    \STATE \hspace{0.6em} Using bounds $(t,h)$ and the code $\mathsf{CSS}[[m,1,d]]$, prepare a pool of encrypted magic states (for all $\mathsf{T/H}$ gates) with traps, and denote the collection by
           $\mathsf{MagicPool}$.
    \STATE \hspace{0.6em} Run
           $(\Gamma_1,\dots,\Gamma_t) \leftarrow \mathsf{GadgetGen}(t,sk_0,\dots,sk_{t-1})$
           to generate all $\mathsf{T}$-gate gadgets.
    \STATE \hspace{0.6em} Send the authenticated evaluation key package to TP
           \[
             \rho_{\mathsf{evk}} \leftarrow
               \mathsf{MAC.Sign}_{k_{\mathsf{TA}}}\big(
                 \mathsf{MAC.Sign}_{k_{\mathsf{TA}}}(\mathsf{ KEY_{eval}},\widetilde{\pi}),\mathsf{MagicPool},
                 \Gamma_1,\dots,\Gamma_t,\widetilde{\sigma}_{\mathsf{aux}}
               \big).
           \]
    \STATE\hspace{0.6em} To obtain the
           secret shares $(sk_t^1,\dots,sk_t^n) \leftarrow \mathsf{SS.Share}(sk_t,n)$ and send the key package $\mathsf{key}_i:=(pk_0,sk_t^i,\pi)$ to $P_i$.
           
    \STATE \textbf{Phase 2: Encryption}
    \STATE For Parties $P_i$: 
    \STATE \hspace{0.6em} Run $k_i\leftarrow \mathsf{MAC.KeyGen}(1^\kappa)$ to obtain the MAC key for $P_i$.
    \STATE\hspace{0.6em} Let $\mathcal{S}_i=\{x^{(i)}_1,\dots,x^{(i)}_{l_i}\}$ and set $L := \max_{i} l_i$.
                     Fix a global ordering $\mathcal{U}=(x_1,\dots,x_L)$ of the universe. Prepare a sequence of $L$ qubits
                     $\mathcal{S}^Q_i=(|{\sigma^i_1}\rangle,\dots,|\sigma^i_L\rangle)$ where
                     \[
                       |\sigma^i_j\rangle=
                       \begin{cases}
                         \ket{1}, & \text{if } x_j\in S_i,\\[0.3em]
                         \ket{0}, & \text{if } x_j\notin S_i.
                       \end{cases}
                     \]
    \STATE\hspace{0.6em} For all according qubit, $P_i$ do $\sigma^{\,i\prime }_j \leftarrow \mathsf{CSS.Encode}(\sigma^i_j)$
    and add traps with permutation $\sigma^{\,i\prime\prime}_j \leftarrow 
            \mathsf{permute}_{\pi}\big(\sigma^{\,i\prime }_j \otimes \ket{0^m}\otimes\ket{+^m}\big)$,   
    and encrypt $\sigma^{\,i\prime\prime\prime}_{j} \leftarrow
            \mathsf{X}^{a^i_j} \mathsf{Z}^{b^i_j}\,\sigma^{\,i\prime\prime}_j\,\mathsf{X}^{a^i_j} \mathsf{Z}^{b^i_j}$ with QOTP keys $(a^i_j,b^i_j)\leftarrow_R\{0,1\}^{3m}$ and return
            \[
            \widetilde{\sigma}^{\,i}_{j} \leftarrow \sigma^{\,i\prime\prime\prime}_{j}\otimes \mathsf{MAC.Sign}_{k_{i}}(\mathsf{HE.Enc}_{pk_0}(a^i_j,b^i_j)).
            \]
    \STATE\hspace{0.6em} Send
            $\widetilde{\mathcal{S}}^Q_i =
              (\widetilde{\sigma}^i_1,\dots,\widetilde{\sigma}^i_L)$ to TP
            (Quantum state together with the classical QOTP keys under MAC).
     \algstore{alg:main}  
  \end{algorithmic}
\end{algorithm*}
\begin{algorithm*}[!tp]\ContinuedFloat
\caption{MP-QPSI (continued)}
\begin{algorithmic}[1]
  \algrestore{alg:main} 
     \STATE \textbf{Phase 3: Homomorphic Evaluation}
     \STATE For Third Party (TP): For clarity of presentation, we suppress the explicit position index $j$ and
describe the evaluation on a generic wire; the same procedure is applied
independently to every position.

     \STATE \hspace{0.6em} Initialise the computation transcript $log$ and fix the quantum circuit $C^{\mathsf{AND}}$ (i.e. $\mathsf{C}^n \mathsf{X}$) that implements the multi-party logical AND.   Decompose it into alternating $c$ Clifford and $t$ $\mathsf{T}$-layers:
             \[
               C^{\mathsf{AND}}
               =
               C^{(1)}_{\mathsf{Cliff}}
               \circ C^{(1)}_{\mathsf{T}}
               \circ C^{(2)}_{\mathsf{Cliff}}
               \circ \cdots \circ
               C^{(t)}_{\mathsf{T}}
               \circ C^{(c)}_{\mathsf{Cliff}} .
             \]
    \STATE \hspace{0.6em} For Clifford layers $C^{(\ell)}_{\mathsf{Cliff}}$: Processes the $\mathsf{TrapTP.EvalCliff}\big(
            C^{(\ell)}_{\mathsf{Cliff}},
            \widetilde\sigma,\widetilde a,\widetilde b,
            \mathsf{KEY}_{eval},\widetilde \pi,log
          \big)$ which applies $C^{(\ell)}_{\mathsf{Cliff}}$ to $\widetilde\sigma$ and,
          via $\mathsf{HE.Eval}$, homomorphically update the encrypted QOTP keys $(\widetilde a,\widetilde b)$ and append this step to $log$.
          \STATE \hspace{0.6em} For $\mathsf{T}$ layers $C^{(\ell)}_{\mathsf{T}}$: Processes the $\mathsf{TrapTP.EvalT}\big(
            C^{(\ell)}_{\mathsf{T}},
            \widetilde\sigma,\widetilde a,\widetilde b,
            \Gamma_{\ell},\mathsf{MagicPool},
            \mathsf{KEY}_{eval},\widetilde \pi,log
          \big)$
          which consumes one gadget
          $\Gamma_{\ell}$ and performs the required Bell measurements through a magic state from $\mathsf{MagicPool}$ to remove the $\mathsf{P}$-correction caused by $\mathsf{T}$ gate, and then uses
          $\mathsf{HE.Eval}$ to homomorphically perform the
          $\mathsf{Recrypt}$ of the QOTP key and intermediate values and randomness are appended to $log$. Conceptually, 
          \[
             \big(\widetilde a_{[pk_{\ell+1}]},\widetilde b_{[pk_{\ell+1}]}\big)
               \leftarrow
               \mathsf{HE.Eval}_{evk_{\ell}}\!\big(
                 \widetilde a_{[pk_{\ell}]},\widetilde b_{[pk_{\ell}]},\mathsf{KEY}_{eval}
               \big)
          \]
          We write subscript $\left [ \cdot  \right ] $ for the encryption under the specific key.
      \STATE \hspace{0.6em} After all Clifford and $\mathsf{T}$-layers of $C^{\mathsf{AND}}$ have been
           processed, the multi–controlled gate $\mathsf{C}^n\mathsf{X}$ implementing
           the logical AND has been applied to the 
           $\widetilde\sigma_{\mathsf{aux}}$ resulting the output $\widetilde\sigma_{\mathsf{out}}$ with $(\widetilde a^{\,\mathsf{out}},\widetilde b^{\mathsf{out}})$ encrypted under $pk_{t}$. 

      \STATE \hspace{0.6em} TP sends the classical data with circuit description
             $\big(\widetilde a^{\,\mathsf{out}},\widetilde b^{\mathsf{out}},log, C^{\mathsf{AND}} \big)$
             to all parties $\{P_i\}$ and $(\widetilde\sigma_{\mathsf{out}},log)$ to  TA.

    \STATE \textbf{Phase 4: Verification and Decryption}
    \STATE \underline{Classical.VerDec:}
    \STATE ($P_i$)\hspace{0.4em}MAC Check: $P_i$ verify its own initial QOTP key ciphertexts
           recorded in $log$: $\mathsf{MAC.Ver}_{k_i}(\widetilde a_i,\widetilde b_i)$.
           if any check fails, set $\mathsf{C.Accept} \gets 0$ and abort.
    \STATE ($P_i$)\hspace{0.6em}Gate array Check: Reconstruct from $log$ the sequence of evaluated gates and compare it with the prescribed gate array $C^{\mathsf{AND}}$. If they differ, set $\mathsf{C.Accept} \gets 0$ and abort.
    \STATE ($P_i$)\hspace{0.6em}Transcript Check: Jointly run the standard transcript checking procedure on $log$ using the public evaluation
           keys $\mathsf{KEY}_{eval}$ and threshold decryption with
           $\{sk_t^i\}_{i\in[n]}$:
           \[
             \mathsf{flag}_{\mathsf{FHE}}
               \leftarrow \mathsf{CheckLog}(\mathsf{KEY}_{eval},log,\{sk_t^i\}). \text{  \quad If  } \mathsf{flag}_{\mathsf{FHE}}=\mathsf{rej}, \text{set } 
           \mathsf{C.Accept} \gets 0 \text{  and abort.}
           \]
    \STATE Otherwise, set $\mathsf{C.Accept} \gets 1$ and the participants obtain the final QOTP keys $(a^{\mathsf{out}}, b^{\mathsf{out}})$ via threshold decryption: \[(a^{\mathsf{out}}, b^{\mathsf{out}}) \leftarrow \mathsf{TFHE.Dec}_{\{sk_t^i\}}(\widetilde a^{\,\mathsf{out}},\widetilde b^{\mathsf{out}})\]
    and send $\mathsf{C.Accept}$ and $(a^{\mathsf{out}}, b^{\mathsf{out}})$ to TA.
    \STATE \underline{Quantum.VerDec:}
    \STATE (TA)\hspace{0.4em}MAC Check: TA verifies the MAC on the evaluation package:
           $\mathsf{MAC.Ver}_{k_{\mathsf{TA}}}(\rho_{\mathsf{evk}})$.
    \STATE (TA)\hspace{0.4em}Trapcode Check: For the $3m$–qubit block, to remove the
             one–time pad:
             \[
               \sigma_{\mathsf{out}} \leftarrow
                 \mathsf{X}^{a^{\mathsf{out}}} \mathsf{Z}^{b^{\mathsf{out}}}\,
                 \widetilde\sigma_{\mathsf{out}}\,
                 \mathsf{X}^{a^{\mathsf{out}}} \mathsf{Z}^{b^{\mathsf{out}}}.
             \]
      Undo the permutation and split target and trap registers:
             \[
               (\sigma_{\mathrm{target}},
                \sigma^{\mathsf{X}\text{-trap}},
                \sigma^{\mathsf{Z}\text{-trap}})
                 \leftarrow
                 \mathsf{permute}_{\pi^{-1}}(\sigma_{\mathsf{out}}).
             \]
      Measure the $\mathsf{X}$–trap block $\sigma^{\mathsf{X}\text{-trap}}$ in the
             computational basis and the $\mathsf{Z}$–trap block $\sigma^{\mathsf{Z}\text{-trap}}$
             in the Hadamard basis; if any outcome is nonzero (resp.\ non–$\ket{+}$),
             then return $(\bot,\mathsf{Accept}\leftarrow 0)$.
      \STATE (TA)\hspace{0.4em}Else, set $\mathsf{Q.Accept} \leftarrow 1$, apply $\mathsf{CSS.Decode}$ to $\sigma_{\mathrm{target}}$ to recover the
logical target qubit, measure it in the according basis, compare the outcome
with the classical bit associated with $\sigma_{\mathsf{aux}}$ to determine the
intersection indicator $b \in \{0,1\}$, and broadcast $b$ to all participants with $\mathsf{Accept}\leftarrow 1$.
  \end{algorithmic}
\end{algorithm*}

\section{Security Analysis}\label{analysis}
\subsection{Correctness}\label{QPSO_Correctness}

To validate the correctness of our logical multi-party AND gate, we instantiate
the quantum subroutine using the multi-controlled gate with \emph{v-chain}
ancilla through the IBM Quantum Platform and the \texttt{qiskit}
library.  In particular, for the three-party case we implement a
three-controlled gate $C_{(3)}^{\mathsf{AND}}$ with controls $q_0,q_1,q_2$ and
target $q_4$, as shown in Fig.~\ref{C3X}.  On the logical level this gate
realises the map:
\[
  C_{(3)}^{\mathsf{AND}} :
  \ket{q_0,q_1,q_2}_Q \ket{0}_{q_4}
  \;\longmapsto\;
  \ket{q_0,q_1,q_2}_Q \ket{q_0\wedge q_1\wedge q_2}_{q_4}
\]
for all $q_0,q_1,q_2\in\{0,1\}$.  We further initialise $(q_0,q_1,q_2)$ in the
uniform superposition $\frac{1}{\sqrt{8}}\sum_{(q_0,q_1,q_2)\in\{0,1\}^3}
\ket{q_0,q_1,q_2}$ and simulate the circuit on \texttt{qiskit}.  The output
distribution of $q_4$, reported in Fig.~\ref{output}, coincides
with the ideal truth table of the logical AND.

The same functional behaviour extends immediately to an $n$-party setting and
to all $L$ positions of the universe. For each position $j$ we therefore use an $n$-controlled gate $C^{\mathsf{AND}}$ with controls $(q_{1,j},\dots,q_{n,j})$ and target $t_j$. In our protocol this logical gate is surrounded by CSS code, trap-code encoding, permutation and QOTP, the whole map is as:
\[
\hspace*{-0.5cm}%
\begin{aligned}
  &\Bigl(\mathsf{CSS.Decode}\circ
    \mathsf{permute}_\pi^{-1}\circ
    \mathsf{QOTP}^{-1}
  \Bigr)
  \circ
  C_j^{\mathsf{AND}}
  \circ\\
  &\Bigl(
    \mathsf{QOTP}\circ
    \mathsf{permute}_\pi\circ
    \mathsf{CSS.Encode}
  \Bigr)
  \bigl(
    \ket{q_{1,j},\dots,q_{n,j}}_{Q_j}\ket{0}_{t_j}
  \bigr)
  \\
  &=\ket{q_{1,j},\dots,q_{n,j}}_{Q_j}
  \ket{\bigwedge_{i=1}^n q_{i,j}}_{t_j}.
\end{aligned}
\]
Moreover, the correctness of the encode layer follows immediately
from the \emph{evaluation correctness} of TFHE and the \emph{correctness} of the MAC
(Definitions~\ref{TFHE.Correctness} and~\ref{MAC.correctness}).
\begin{figure*}[t]
    \centering
    \includegraphics[width=0.7\linewidth]{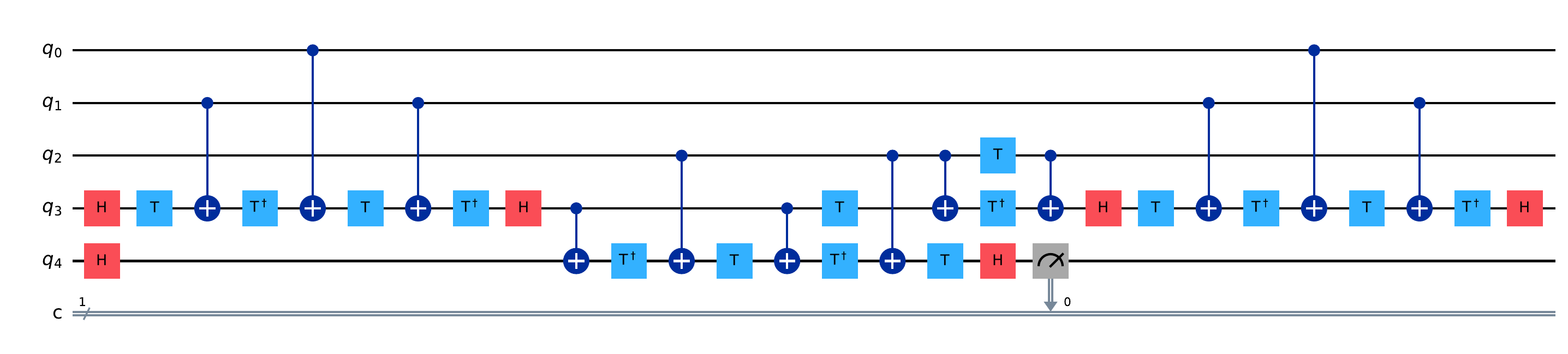}
    \caption{Decomposition of the three-controlled AND with v-chain ancilla.}
    \label{C3X}
\end{figure*}
\begin{figure*}[t]
    \centering
    \includegraphics[width=0.7\linewidth]{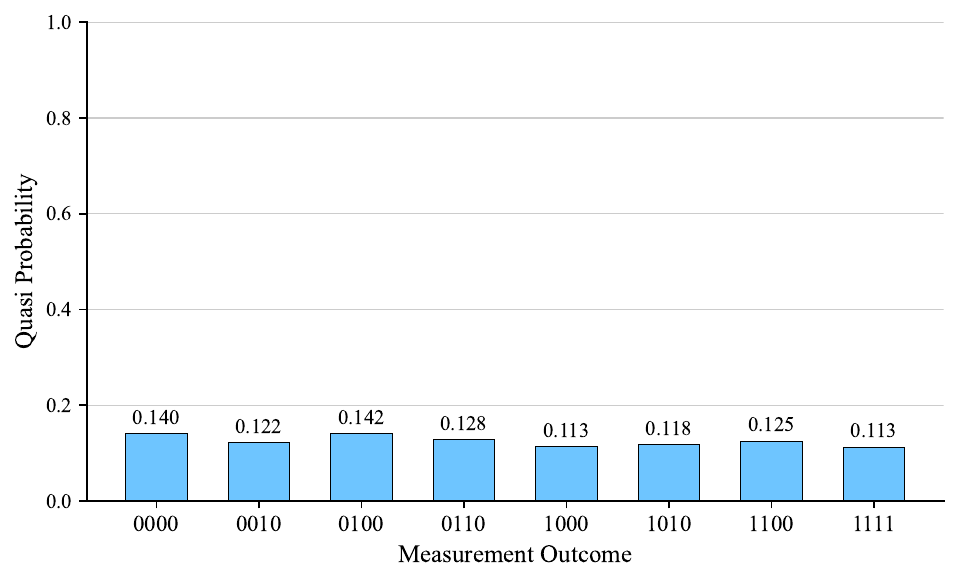}
    \caption{The Measurement outcome of $q_4$ together with the $q_0,q_1,q_2$ in uniform superposition.}
    \label{output}
\end{figure*}
\subsection{Participant Privacy}\label{Privace}
In this section, we will evaluate the privacy of the participants against TP, outside eavesdropper and other participants as follows:
\begin{enumerate}[label=(\arabic*)]
  \setlength{\itemsep}{0pt}
  \setlength{\parskip}{0pt}
  \setlength{\topsep}{2pt}
    \item TP learns nothing about the private sets of $P_i$;
    \item Any outside eavesdropper is unable to learn anything about the private sets of $P_i$;
    \item $P_i$ cannot obtain extra information about the private sets of other participants beyond the intersection, even when conspiring with other participants or TP.
\end{enumerate}
\begin{theorem}[Privacy against TP]\label{P_against_TP}
\emph{In the MP-QPSI scheme proposed in Protocol~\ref{MP-QPSI},  TP cannot obtain any information about the private sets of participants.}
\end{theorem}

\begin{proof}
\emph{Quantum part.}
According to the encryption phase of Protocol~\ref{MP-QPSI}, from the point of view of the TP who receives quantum state under QOTP:
\[
  \tilde\sigma
  =\mathsf{X}^a \mathsf{Z}^b\,
  \mathsf{permute}_\pi(\mathsf{CSS.Encode}(\sigma)\otimes|0^m\rangle\otimes|+^m\rangle)\,
  \mathsf{Z}^b \mathsf{X}^a .
\]
where $(a,b)\in\{0,1\}^{3m}$, According to the property and average over the uniformly random key yields TP's reduced density matrix
\[
  \rho^{(\sigma)}_{\mathrm{TP}}
  =
  \frac{1}{4^{3m}}
  \sum_{a,b} \tilde\sigma
  = \frac{\mathbb{I}}{2^{3m}},
\]
which is maximally mixed and information-theoretically secure. In particular, $\rho^{(0)}_{\mathrm{TP}}=\rho^{(1)}_{\mathrm{TP}}$, so the success probability of any Unambiguous State Discrimination (USD) test \cite{Herzog05} is negligible which means TP cannot extract any information on the logical bits from the quantum registers.

\emph{Classical part.}
All input-dependent classical data are encrypted under TFHE and authenticated by MAC. Combining the IND-CPA security of TrapTP and TFHE with Lemma 1 in \cite{Dulek16}, and EUF-CMA security of the MAC, all classical ciphertexts and tags observed by TP are computationally indistinguishable. 
\end{proof}

\begin{theorem}[Privacy against eavesdroppers]\label{P_against_Outside}
\emph{Any external adversary $\mathcal{E}$ that can wiretap both the quantum and classical channels and perform arbitrary quantum attacks still learns no information about the private sets.}
\end{theorem}

\begin{proof}
Same as mentioned in Theorem~\ref{P_against_TP}, each logical qubit is first encoded under CSS, traps, permutation and QOTP. Thus an outside eavesdropper without the QOTP keys and permutation obtains no information about the private
sets due to the fact that the quantum states are maximally mixed, and the classical information are encrypted under TFHE and authenticated by MAC.

Moreover, our encoding makes the protocol against standard outside attacks such as intercept-measure-resend, measurement-resend, dense-coding attacks, and auxiliary-qubit attacks. \cite{Hoi05} For example, any attempt from $\mathcal{E}$ to intercept and re-prepare the sequence without knowing the trap positions and their correct measurement bases will inevitably introduce errors which will be detected in the $\mathsf{VerDec}$ procedure, and post-interaction auxiliary qubit can not reveal meaningful information due to the maximally mixed states.
\end{proof}

\begin{theorem}\label{P_against_Pi} \emph{Let $\mathcal{P}_C$ be any subset of corrupted participants with $|\mathcal{P}_C|<t$.
Even if the parties in $\mathcal{P}_C$ collude with the TP, they obtain no information about the honest parties' private sets beyond the intersection.}
\end{theorem}
\begin{proof}
(1) \emph{Intercepting honest channels to TP.}
The coalition $\mathcal{P}_C$ may try to intercept all classical and quantum messages sent by honest parties to TP and reconstruct their private sets. Classically, honest parties only send TFHE ciphertexts of QOTP keys together with MAC tags. Since the TFHE secret key $sk$ is shared with a $n_t$-out-of-$n$ SS scheme like Definition~\ref{ss}, The privacy of SS in Definition~\ref{SS.Privacy} and Lemma~\ref{tfhe-basic} implie that the shares held by $\mathcal{P}_C$ (with $|P_C|<t$) reveal no information about $sk$.  On the quantum side, without the corresponding QOTP keys the intercepted states are exactly the QOTP-masked codewords which is maximally mixed state and information-theoretically secure.

(2) \emph{Colluding with TP to deviate in the evaluation.}
The coalition $\mathcal{P}_C$ may also collude with TP and ask TP to run a modified
evaluation to extract extra information (e.g., a different gate array) and try to pass
$\mathsf{Classical.VerDec}$ procedure.  Members of $\mathcal{P}_C$ can reveal their own MAC keys
$k_i$ to TP to pass the local MAC checks and gate array check of corrupted side, but MAC keys of honest participants remain unknown,  any deviation that affects ciphertexts or logs will cause some MAC or gate array check from honest side to fail. We can set a threshold number of necessary passes for the local checks to ensure that the whole $\mathsf{Classical.VerDec}$ process detects and terminates the aforementioned corruption with an overwhelming probability.

In both cases, the joint view of $\mathcal{P}_C$ can be simulated from its own inputs and the intersection alone, hence $\mathcal{P}_C$ gains no extra information about the other
participants' private sets.
\end{proof}
\begin{figure*}[t]
    \centering
    \includegraphics[width=0.55\linewidth]{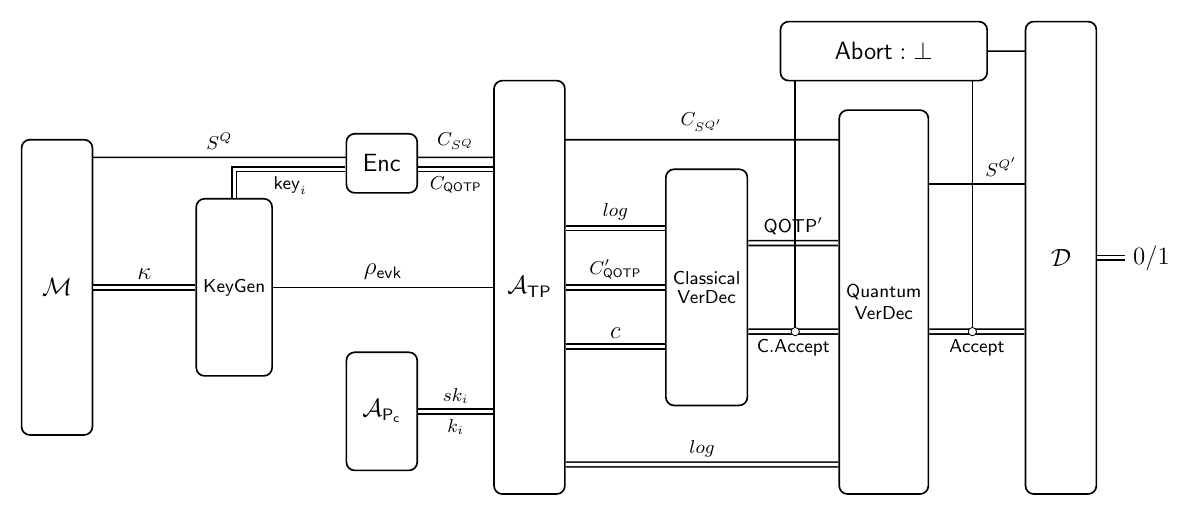}%
    \hfill
    \includegraphics[width=0.45\linewidth]{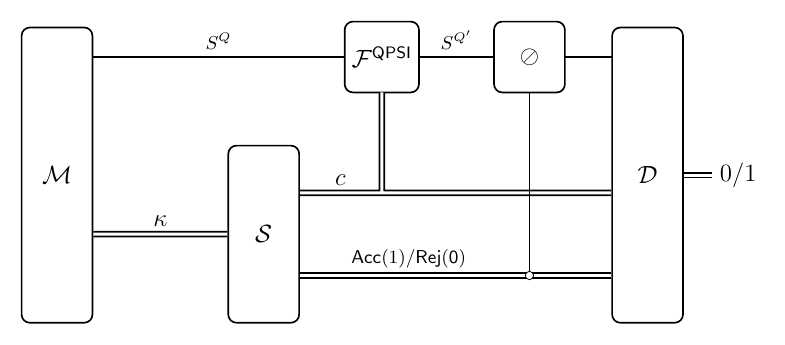}
    \caption{
        The Real-world (left) and Ideal-world (right) model.
    }
    \label{model}
\end{figure*}
\begin{figure*}[t]
    \centering
    \includegraphics[width=0.8\linewidth]{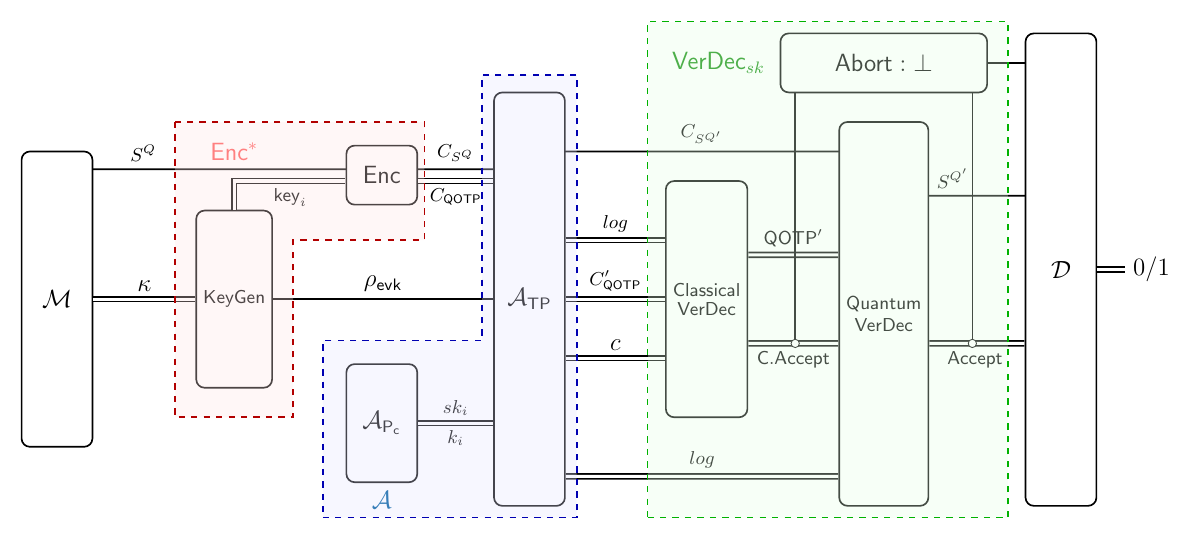}
    \caption{Modularly encapsulated hybrid of $H_0$.}
    \label{real1}
\end{figure*}
\subsection{Verifiability Against Malicious TP}\label{Verifiablity}
We now turn to the verifiability of our protocol against a malicious TP. At a high level, the security of a \textsf{vQFHE} scheme requires that any server who deviates significantly from the prescribed evaluation map $\mathsf{vQFHE.Eval}$ will be caught by the verification procedure: whenever the algorithm $\mathsf{VerDec}$ accepts, the decrypted output state must be close to the ideal output obtained by applying the target circuit to the input.

Following the $\kappa$-SEM-VER notion in Definition~\ref{sem-ver}, together with the quantum authentication~\cite{Dupuis12,Broadbent16} and semantic secrecy~\cite{Alagic16}, we define a Real experiment and an Ideal experiment for our MP-QPSI protocol as depicted in Fig.~\ref{model}.  Informally, the figure captures that every adversary $\mathcal{A}$ with access to the ciphertexts in the Real world can be simulated by a simulator $\mathcal{S}$ that has access
only to an ideal PSI functionality in the Ideal world.  In the remainder of this section, we \emph{encapsulate} certain components of the Real model into \emph{functionality blocks} and use \emph{hybrid chain} to prove the \hyperref[ver]{verifiability} that for any quantum poly-time distinguisher
$\mathcal{D}$ \label{ver}:
\[
  \bigl|
    \Pr\bigl[\mathcal{D}(\mathsf{Real}_{\mathcal{A}}^{\Pi^{\mathsf{MP-QPSI}}})=1\bigr]
    -
    \Pr\bigl[\mathcal{D}(\mathsf{Ideal}_{\mathcal{S}}^{\mathcal{F}^{\textsf{QPSI}}})=1\bigr]
  \bigr|
  \leq \mathsf{negl}(\kappa),
\]
which formalises that the Real and Ideal executions of our protocol are
computationally indistinguishable and the protocol is verifiable in the semantic sense.\\

\noindent$\mathbf{H_0}$: This is the Real model $\mathsf{Real}_{\mathcal{A}}^{\Pi^{\mathsf{MP-QPSI}}}$ as shown in Fig.~\ref{model}, which is designed based on the interaction process of our proposed MP-QPSI protocol \ref{MP-QPSI}. \\

\noindent$\mathbf{H_1}$: From the Real experiment $\mathrm{H}_0$, we first pass to $\mathrm{H}_1$ by encapsulating the \textsf{KeyGen} and \textsf{Enc}
modules into a single functionality $\mathsf{Enc}^{*}$ as shown in the red areaof Fig.~\ref{real1}. For a subset of honest parties $\mathcal{P}_S$, with $|S|\ge t$. By the correctness of SS (Definition~\ref{SS.Correctness}), from the viewpoint of the adversary we equivalently treat “TA plus the honest
participants” as a single honest entity that locally holds $sk$,
just as in the original TrapTP setting:
\[
  \mathsf{view}_{\mathcal{A}}^{\mathrm{H_0}}
  \bigl( \{sk_i\}_{i\in S}, \,\cdot\,  \bigr)
  \;\equiv\;
  \mathsf{view}_{\mathcal{A}}^{\mathrm{H_1}}
  \bigl( sk ,\,\cdot\, \bigr)
\]
which lead to $\mathsf{AdvHyb}_{\mathrm{H_0}}^{\mathrm{H_1}}(\mathcal{A}) = 0$ where $\mathsf{AdvHyb}$ is used as distinguishing program between hybrid models.
\\

\noindent$\mathbf{H_2}$: In Fig.~\ref{real1}, the blue region corresponds to the TP together with a subset of colluding participants $\mathcal{P}_C$. In hybrid $\mathrm{H}_2$ we conceptually bundle these components
into a single unified adversary $\mathcal{A}=(\mathcal{A}_{\mathsf{TP}},\mathcal{A}_{\mathcal{P}_C})$ who controls both the TP and the corrupted parties.  By Theorem~\ref{P_against_Pi} and SS privacy (Definition~\ref{SS.Privacy}), any information about the honest private sets that can be obtained jointly by $\mathcal{A}_{\mathsf{TP}}$ and $\mathcal{A}_{\mathcal{P}_C}$ can be closely simulated from the view of the unified adversary $\mathcal{A}$. For $\mathsf{AdvHyb}_{\mathrm{H}_1}^{\mathrm{H}_2}(\mathcal{A})$:
\[
  \Bigl|
    \Pr\bigl[\mathcal{A}(\mathsf{view}^{\mathrm{H}_1}_{\mathcal{A}}(\,\cdot\,))=1\bigr]
    -
    \Pr\bigl[\mathcal{A}(\mathsf{view}^{\mathrm{H}_2}_{\mathcal{A}}(\,\cdot\,,\{k_i,sk_i\}_{i\in C}))=1\bigr]
  \Bigr| 
\]
with probability $\mathsf{negl (\kappa)}$. So, treating TP and the colluding participants as a single adversarial interface does not alter the distribution of the adversary's view with overwhelming probability.\\

\noindent$\mathbf{H_3}$: In $\mathrm{H}_3$, we encapsulate the green part of Fig.~\ref{real1}, namely \textsf{Classical.VerDec}, \textsf{Quantum.VerDec} and the \textsf{Abort} box, into a single verification algorithm $\mathsf{VerDec}_{sk}$.   The critical point in this encapsulation step is whether the distributed verification process in our scheme can achieve the same level of verifiability as $\mathsf{VerDec}_{sk}$ in $\mathsf{vQFHE}$. That is, any deviation from the target circuit will be rejected at this stage. 
\begin{lemma}\label{LH3}
    For any QPT adversary $\mathcal{A}$,  there exists adversaries $\mathcal{B}_1$ and $\mathcal{B}_2$ such that:
    \[
    \mathsf{AdvHyb}_{\mathrm{H}_2}^{\mathrm{H}_3}(\mathcal{A}) \le \mathsf{Adv}_{\mathsf{MAC}}^{\mathsf{euf-cma}}(\mathcal{B}_1)+\mathsf{Adv}^{\mathsf{Trap}}(\mathcal{B}_2)+\mathsf{negl}(\kappa).
    \]
\end{lemma}
\begin{proof}
In our protocol the verification is combined as
\[
  \mathsf{Accept}=1
  \;\Longleftrightarrow\;
  \mathsf{C.Accept}=1 \;\land\; \mathsf{Q.Accept}=1,
\]
where $\mathsf{C.Accept}$ is the outcome of \textsf{Classical.VerDec} (MAC check, gate array check and log check) and $\mathsf{Q.Accept}$ is the outcome of \textsf{Quantum.VerDec}
(trapcode check). By correctness and EUF–CMA security of the MAC in Definition~\ref{MAC.correctness}, evaluation correctness of TFHE in Definition~\ref{TFHE.Correctness}, and soundness of the trap code. We can directly derive that the probability:
\[
  \Pr[\mathsf{Accept}=1\wedge S^{Q'} \neq \Phi_{C^{\mathsf{AND}}}(S^{Q})]
\]
is smaller than $\mathsf{Adv}_{\mathsf{MAC}}^{\mathsf{euf-cma}}(\mathcal{B}_1)+\mathsf{Adv}^{\mathsf{Trap}}(\mathcal{B}_2)+\mathsf{negl}(\kappa)$ where they represent the maximal advantages of any QPT adversaries $\mathcal{B}_1,\mathcal{B}_2$ in breaking the MAC check and the trapcode check, and the negligible errors $\mathsf{negl}(\kappa)$ come from two parts: $\mathsf{negl}_\mathsf{HE}(\kappa)$ comes from evaluation of $\mathsf{HE}$ in log check and gate array check and $\mathsf{negl}_\mathsf{CSS}(\kappa)$ comes from the decoding failure \cite{G09}.
\end{proof}
According to the Lemma above, the joint map of \textsf{Classical.VerDec}, \textsf{Quantum.VerDec} and \textsf{Abort} is computationally indistinguishable from an ideal single verification routine $\mathsf{VerDec}_{sk}$.\\

\noindent$\mathbf{H_4}$: After encapsulating the above three blocks, the resulting hybrid model is in fact an instance of the \emph{semantic real} model of Definition~5 in ~\cite{vQFHE17} where the \textsf{vQFHE} is a combination of TrapTP and TFHE, instantiated with our logical circuit $C^{\mathsf{AND}}$.  All structural requirements of that definition are met: our decryption algorithm belongs to $\mathsf{NC}^1$ and the circuit for $C^{\mathsf{AND}}$ has a polynomially bounded number of $\mathsf{T}$ gates. Moreover, Theorem~5 in~\cite{vQFHE17} already shows that TrapTP is $\kappa$-IND-VER, and hence $\kappa$-SEM-VER. Applying the definition of $\kappa$-SEM-VER in Definition~\ref{sem-ver} to $\mathrm{H}_3$, we can replace the semantic real model by the corresponding ideal model.  This yields exactly the ideal functionality $\mathsf{Ideal}_{\mathcal{S}}^{\mathcal{F}^{\textsf{QPSI}}}$ shown in Fig.~\ref{model} and lead to $\mathsf{AdvHyb}_{\mathrm{H}_3}^{\mathrm{H}_4}\le \mathsf{negl}(\kappa)$ .

As mentioned above, we can finally bound $\mathsf{AdvHyb}_{\mathrm{H}_0}^{\mathrm{H}_4}$ by a negligible function in \(\kappa\) based on the hybrid reduction derived from the preceding intermediate analysis and lead to the \hyperref[ver]{verifiability} of our scheme mentioned above.

\section{Flexible and Modular Realizations of MP-QPSO}\label{Framework}
The verifiability analysis in Section~\ref{Verifiablity} already suggests a modular interpretation of our construction.  At a high level, the MP-QPSI protocol against malicious TP and collusion can be seen as the composition of a small number of black-box modules: (i) a threshold layer makes the scheme resistant to collusion of the participants; (ii) a vQFHE layer against any deviation from the target circuit; (iii) a logical quantum circuit implementing the desired set operation. Consequently, different circuit instantiations or alternative QHE constructions can be plugged into the same framework to achieve different effects. 

In the remainder of this section we illustrate the modular nature of our framework through several examples, including choosing alternative intersection circuits to reduce the number of \(\mathsf{T}\)-gates, modifying the circuit to realize QPSU, and replacing the underlying QHE logic to obtain gadget-free instantiations.
\subsection{Circuit-Level Variants for Intersection}
\begin{figure*}[t]
    \centering
    \includegraphics[width=0.7\linewidth]{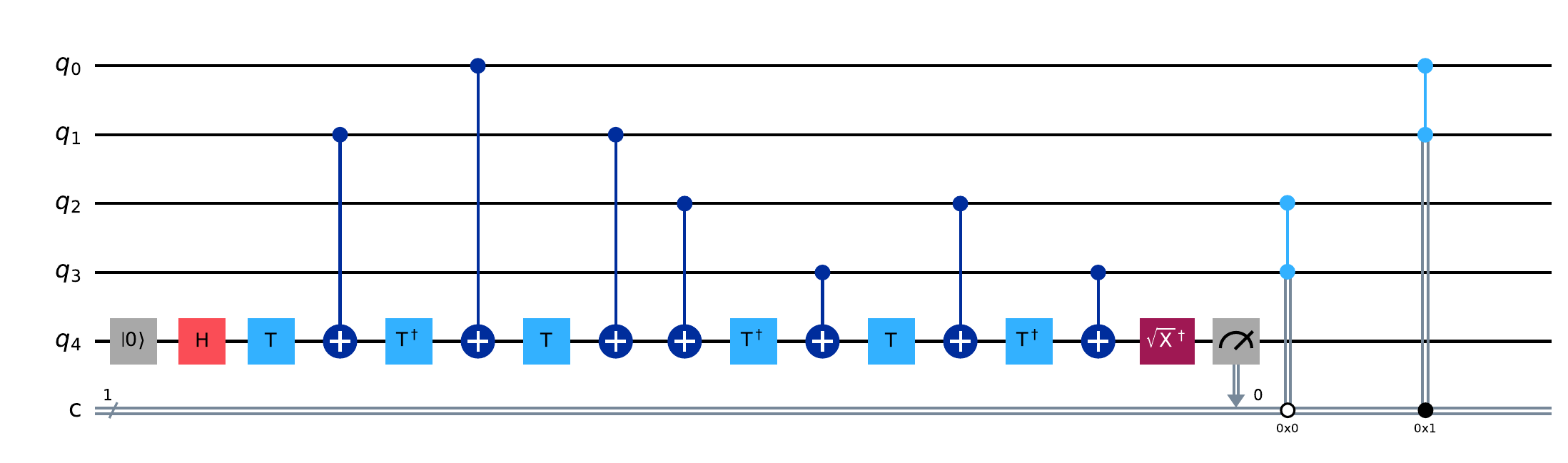}
    \caption{CCCZ circuit built using only 6 T gates.}
    \label{C3z}
\end{figure*}
\emph{CCCZ-based reduction of \,\(\mathsf{T}\)-gates.} In our basic instantiation, the multi-party AND circuit \(C^{\mathsf{AND}}\) is  a \(\mathsf{C}^n \mathsf{X}\) with v-chain ancilla which costs approximately \(7(n-1)\) \(\mathsf{T}\)-gates for \(n\) inputs.  By instead expressing the intersection as a multi-controlled phase and using the CCCZ construction of~\cite{GC21} (we have simulated the realization in \texttt{qiskit} as Fig.~\ref{C3z}), we can implement an \(n\)-controlled AND with only \(4n-6\) \(\mathsf{T}\)-gates.  Thus, within the same MP-QPSI framework, replacing the original \(\mathsf{C}^n \mathsf{X}\) block by a CCCZ-based realization yields a strictly lower \(\mathsf{T}\)-gate complexity for the homomorphic evaluation phase.

\emph{QuAND realization.} When the TP is equipped with strong quantum capabilities (e.g., a high-scalability superconducting processor), the logical AND subroutine in our intersection circuit can be instantiated using the QuAND gate \cite{Chu21}. A QuAND can be synthesized by combining a single-qubit $X$ gate with an iSWAP-like operation that selectively swaps $\ket{11}\leftrightarrow\ket{20}$, exploiting the $\ket{2}$ ancilla level only for temporary state transfer while maintaining reversibility. Under this setting, \cite{Chu21} have experimentally realized the resulting $n$-qubit controlled-$Z$ (and thus a generalized Toffoli) on an 8-qubit superconducting platform; when compiled onto a 1D chain, it requires $2n-3$ two-qubit-gate cycles and achieves overall depth $\mathcal{O}(n)$.
\subsection{From QPSI to QPSU via Open Controls}\label{SQPSO}
In our MP-QPSI construction, the predicate bit is obtained by the circuit \(C^{\mathsf{AND}}\). To realize QPSU within exactly the same framework, we only need to change this predicate from an AND to an OR.  For Boolean inputs
\(x_1,\dots,x_n \in \{0,1\}\), we use De Morgan’s law
\[
  \operatorname{OR}(x_1,\dots,x_n)
  \;=\;
  \bigvee_{i=1}^n x_i
  \;=\;
  \neg\!\left(\,\bigwedge_{i=1}^n \neg x_i\,\right).
\]
At the circuit level this is implemented by turning the multi–controlled \(\mathsf{C}^n \mathsf{X}\) into an multi–open–controlled gate: each control qubit is conjugated by an \(\mathsf{X}\) gate before and after the control place as shown in Fig.~\ref{OR} (realizing control on \(\ket{0}\) instead of \(\ket{1}\)), and the target qubit is flipped once more to
account for the outer negation. 
\Needspace{\baselineskip}
\begin{center}
  \includegraphics[width=0.9\columnwidth]{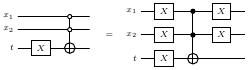}
  \captionof{figure}{Multi-Open-Controlled X Gate}
  \label{OR}
\end{center}
\subsection{Alternative vQHE Instantiations}\label{AltvQFHE}
The verifiability of the underlying vQFHE scheme is what gives our MP-QPSI protocol its distinguishing ability to detect a TP that deviates from the prescribed circuit.  In Protocol~\ref{MP-QPSI} we instantiate this layer with TrapTP~\cite{vQFHE17}, the first vQFHE construction, whose complicate encoding and trap structure yield very strong security guarantees and a clean \(\kappa\)-SEM/IND-verifiability analysis.  This comes at the price of higher complexity: the quantum communication of TrapTP is on the order of \(\mathcal{O}(18n)\).

Within our modular framework, the TrapTP component can be replaced by more communication-efficient vQFHE schemes.  For example, Broadbent’s information-theoretic secure construction~\cite{Broadbent18} reduces the quantum communication to \(\mathcal{O}(6n)\) at the cost of additional interaction between client and server~\cite{LW25}, while more recent schemes such as~\cite{He24,LW25} achieve \(\mathcal{O}(6n)\) quantum communication together with essentially non-interactive verification. All these vQHE variants can be plugged into the MP-QPSO template and combined with the same TFHE-based threshold layer to obtain multi-party, verifiable QPSI and its extensions.

\section{Performance analysis}\label{Performance}
\emph{Quantum Communication Cost.} Each party encodes its $L$-qubit universal private set sequence like the TrapTP scheme \cite{vQFHE17}. Specifically, each block is encoded via a $\textsf{CSS}[[m,1,d]]$ code and augmented with $2m$ trap qubits ($\ket{0}$ and $\ket{+}$), resulting in a $3mL$-qubit authenticated block. The total communication is thus $3nmL$ qubits for $n$ parties. Considering the auxiliary information from TA (pre-shared EPR pairs for $t$ and $h$ gates and auxiliary quantum information). Due to the feasibility constraints of TrapTP instantiations, the parameters satisfy $m, t, h = \mathcal{O}(\kappa)$, the aggregate quantum communication complexity is $\mathcal{O}(n\kappa L)$. 

\emph{Quantum Computation Cost.} TP homomorphically evaluates the multi-party intersection circuit through $3nmL$ qubits. The underlying logical quantum computation consists of $\mathcal{O}(nmL)$ gate operations. On the verifier side, TA measures $2mL$ trap qubits and performs the corresponding decoding checks, which costs $\mathcal{O}(mL)$. With $m=\mathcal{O}(\kappa)$, the overall quantum computation cost is dominated by TP and can be summarized as $\mathcal{O}(n\kappa L)$.
\begin{table*}[t]
  \centering
  \small
  \caption{Comparison of quantum private set protocols.}
  \label{comparison}
  \hspace*{-4mm}
  \begin{tabular}{ccccccccc}
    \toprule
    Protocol
      & \begin{tabular}[c]{@{}c@{}}Quantum\\Resources\end{tabular}
      & \begin{tabular}[c]{@{}c@{}}Quantum\\Technologies\end{tabular}
      & \begin{tabular}[c]{@{}c@{}}Party\\Scenarios\end{tabular}
      & \begin{tabular}[c]{@{}c@{}}Supported\\Functionality\end{tabular}
      & \begin{tabular}[c]{@{}c@{}}TP\\Model\end{tabular}
      & \begin{tabular}[c]{@{}c@{}}Collusion\\Resistance\end{tabular}
      & \begin{tabular}[c]{@{}c@{}}Communication\\Complexity\end{tabular}\\
    \midrule
    \cite{QPSI_Shi15}  & Encoded States & \begin{tabular}[c]{@{}c@{}}Complicated\\operator G\end{tabular} & Two-Party & PSI & Semi-honest & N/A & $\mathcal{O}(L)$  \\
    \cite{QPSI_Zhang20}& GHZ State & \begin{tabular}[c]{@{}c@{}}Pauli\\operation\end{tabular} & Three-Party & \begin{tabular}[c]{@{}c@{}}PSI-CA\\PSU-CA\end{tabular} & Semi-honest & Client & $\mathcal{O}(L)$  \\
    \cite{QPSI_D21}    & Single Photons & \begin{tabular}[c]{@{}c@{}}Asymmetric key\\distribution\end{tabular} & Two-Party & PSI & Semi-honest & N/A & $\mathcal{O}(L+\kappa)$  \\
    \cite{QPSI_Liu21}  & Single Photons & \begin{tabular}[c]{@{}c@{}}Pauli\\operation\end{tabular} & Two-Party & PSI-CA & Semi-honest & N/A & $\mathcal{O}(LlogL)$  \\
    \cite{QPSI_M23}    & Single Photons & \begin{tabular}[c]{@{}c@{}}Pauli\\operation\end{tabular} & Multi-party & PSI & N/A & Single client & $\mathcal{O}(nL)$  \\
    \cite{QPSI_Liu23}    & Single Photons & QHE & Two-party & PSI-CA & Semi-honest & N/A & $\mathcal{O}(L)$  \\
    \cite{QPSI_Chi24}  & Bell States & \begin{tabular}[c]{@{}c@{}}Qubit\\operation\end{tabular} & Three-party & \begin{tabular}[c]{@{}c@{}}PSI\&PSU\\Mixed-CA\end{tabular} & Semi-honest & None & $\mathcal{O}(L)$  \\
    \cite{QPSI_M24}    & Single Photons & \begin{tabular}[c]{@{}c@{}}Rotation\\operation\end{tabular} & Two-Party & Threshold PSI & Trusted & N/A & $\mathcal{O}(\kappa L)$  \\
    \cite{QPSI_Huang24}& Single Photons & \begin{tabular}[c]{@{}c@{}}Rotation\\operation\end{tabular} & Multi-party & PSI & Semi-honest & Client & $\mathcal{O}(n\kappa L)$  \\
    Ours               & Encoded States & vQFHE & Multi-Party & PSI\&PSU & Malicious & TP–client & $\mathcal{O}(n\kappa L)$ \\
    \bottomrule
  \end{tabular}
\end{table*}

\emph{Performance Comparison.} The performance of our proposed MP-QPSO, compared with other state-of-the-art protocols, is summarized in Table~\ref{comparison}. The comparison is conducted across several dimensions, including the quantum resources, quantum technologies, party scenarios, supported functionalities (e.g., PSI/PSU or only cardinality), the security level against the TP and parties (collusion), and the communication complexity. Since our Protocol~\ref{MP-QPSI} is designed for the multi-party setting and can be straightforwardly extended to MP-QPSU via the open-control described in Section~\ref{SQPSO}. More importantly, our framework provides the key distinguishing ability of resisting malicious behaviors from the TP, including conspiring between the TP and a threshold number $n_t$ of parties. 

\emph{Party scenarios}: From Table~\ref{comparison}, we observe that \cite{QPSI_Shi15,QPSI_D21,QPSI_Liu21,QPSI_Liu23,QPSI_M24} are restricted to the two-party setting, typically instantiated in a client-server form. The protocols in \cite{QPSI_Zhang20,QPSI_Chi24} target the three-party setting, whereas only \cite{QPSI_M23,QPSI_Huang24} and our proposal support the multi-party scenario; \emph{Functionality}: None of the existing schemes can compute both the intersection and union (PSI\&PSU) as our protocol does. Specifically, \cite{QPSI_Shi15,QPSI_D21,QPSI_M23,QPSI_M24,QPSI_Huang24} only support PSI, where \cite{QPSI_M24} further considers threshold PSI that reveals the intersection only when it is large enough. \cite{QPSI_Zhang20} supports the cardinality of intersection and union (PSI-CA\&PSU-CA), while \cite{QPSI_Liu21,QPSI_Liu23} provide only PSI-CA. \cite{QPSI_Chi24} focuses on a specialized mixed-cardinality functionality like $\left | (A \cup B) \cap C \right |$; \emph{Security}: among the protocols that rely on TP, most adopt a semi-honest TP assumption, although some schemes have attack-resistance or detection mechanisms (e.g., \cite{QPSI_Shi15,QPSI_D21,QPSI_Huang24}) and most schemes do not address collusion. Only \cite{QPSI_Zhang20,QPSI_Huang24} against collusion which is limited to party-only, these works emphasis that the TP cannot conspire with any party which further highlights another advantage of our MP-QPSO framework: it prevents conspiring between the TP and fewer than a threshold number of parties which is enabled by TFHE.

\section{Conclusion}\label{Conclusion}
\emph{Conclusion.} In this work, we present Protocol~\ref{MP-QPSI}, a MP-QPSI construction that extends vQFHE to the multi-party setting via TFHE and instantiates the intersection functionality using a $C^{\mathsf{AND}}$ circuit, accompanied by the corresponding quantum circuits and simulations. We then provide a comprehensive analysis of correctness and participant privacy against TP, eavesdroppers, and collusive behaviors, and further establish verifiability against a malicious TP under the semantic security model. Adopting a framework-level view, we additionally present several modular realizations of the proposed construction, and show that it can be extended to MP-QPSO via the open-control technique. Finally, we compare our framework with a broad range of existing QPSO schemes, and include Table~\ref{comparison} to highlight our distinguishing advantages in terms of the TP model and collusion resistance: a mechanism to detect TP deviations from the prescribed circuit, as well as resilience against collusion between the TP and a subset of participants.

\emph{Future work.} Although Section~\ref{Performance} and Table~\ref{comparison} indicate that our protocol incurs only a moderate quantum cost, the strong security guarantees inherited from TrapTP~\cite{vQFHE17} lead to a substantially higher classical overhead~\cite{LW25}. This classical–quantum trade-off is largely inherent to QHE-based constructions. Moreover, handling \textsf{T} gates in the evaluated quantum circuit typically requires additional gadget, which increases implementation complexity. Accordingly, our future work will focus on the following directions: (1) leveraging encrypted \textsf{CNOT} operations based on trapdoor claw-free function pairs to avoid gadget-based \textsf{T}-handling and to achieve $\mathcal{O}(2n)$ quantum communication complexity~\cite{Mahadev18}; (2) redesigning the $C^{\mathsf{AND}}$ circuit to completely eliminate the complexity from \textsf{T} gates; and (3) exploiting the presence of the TA to further relax the quantum capability requirements imposed on all participants.

\bibliographystyle{alpha3}  
\bibliography{myrefs}       
\end{document}